\documentclass[9pt,conference,english]{IEEEtran}
\IEEEoverridecommandlockouts

\usepackage[a4paper,top=1.9cm,inner=1.3cm,outer=1.3cm,bottom=2.9cm]{geometry}
\usepackage[cmex10]{amsmath}
\usepackage[noadjust]{cite}
\usepackage{%
    amssymb,%
    amsthm,%
    babel,%
    color,%
    enumerate,%
    etoolbox,%
    mathtools,%
    stmaryrd%
}
\usepackage{authblk}
\usepackage{tikz,pgf,pgfplots}
\usepgflibrary{shapes}
\usetikzlibrary{%
    arrows,%
    decorations.text,%
    positioning,%
    scopes,%
    shapes%
}
\theoremstyle{definition}

\theoremstyle{theorem}

\theoremstyle{prop}
\newtheorem{prop}{Proposition}

\theoremstyle{lemma}

\theoremstyle{remark}

\begin{document}

\title{Queuing Theoretic Models for Multicast and Coded-Caching in Downlink Wireless Systems}
\author[1]{Mahadesh Panju} 
\author[1]{Ramkumar Raghu} 
\author[1]{Vinod Sharma}  
\author[2]{Rajesh Ramachandran} 
\affil[1]{Indian Institute of Science, Bangalore, INDIA. \textit{\{mahadesh,ramkumar,vinod\}@iisc.ac.in}}
\affil[2]{CABS, DRDO, Bangalore, India. \textit{rajesh81@gmail.com}}
\maketitle
	\begin{abstract}
We consider a server connected to $L$ users over a shared finite capacity link. Each user is equipped with a cache. File requests at the users are generated as independent Poisson processes  according to a popularity profile from a library of $M$ files. The server has access to all the files in the library. Users can store parts of the files or full files from the library in their local caches. The server should send missing parts of the files requested by the users. The server attempts to fulfill the pending requests with minimal transmissions exploiting multicasting and coding opportunities among the pending requests.\\
\indent We study the performance of this system in terms of queuing delays for the naive multicasting and several coded multicasting schemes proposed in the literature. We also provide approximate expressions for the mean queuing delay for these models and establish their effectiveness with simulations.
\end{abstract}

\section{Introduction}
Future wireless networks (e.g., 5G) are expected to deliver high volumes of data (\cite{CVNI}). Video on demand accounts for a major fraction of the network traffic. The current network architecture cannot scale cost effectively to meet the exploding demands. However, in recent years storage has become inexpensive. Also, a small number of contents constitute a major proportion of the traffic \cite{itube}. These two factors suggest the use of caching content close to the users as one of the potential solutions for meeting user demands. This is the topic of discussion in this paper.\\
\indent There are two main approaches of caching studies currently considered in the literature. In the first approach (conventional) (\cite{podlipnig,Martina2016}), the focus is on considering the eviction policies at individual caches, such as First-In-First-Out (FIFO), Least-Frequently-Used (LFU), Least-Recently-Used (LRU), Time-To-Live (TTL) etc. where maximizing the hit probability in a local cache  is the parameter of primary interest. In \cite{Fofack2012}, a hierarchical network based on TTL caches is analyzed. The general cache network is studied in \cite{Fofack2014}. In \cite{Martina2016}, the performance of the different caching policies under independent reference model (IRM) traffic and renewal traffic is analyzed. Another group of the works considers caching in a wireless heterogeneous network. The works in  \cite{Poularakis2014,Cui2017,Yang2016} study  caches in small cell networks. Some works (\cite{Ji2015,Zhang2016}) consider caching in the context of Device-Device (D2D) communication. \\
\indent The second, more recent approach (\cite{Maddah-Ali2014,Ji2016}) considers static caches sharing a common link with a server. These systems have two different phases: content placement  and coded delivery. In the content placement phase caches are populated (usually under low network activity) with files. In the content delivery phase, requests from all the nodes are received by the server and the delivery is performed using \textit{coded multicast}. This has been shown to reduce the total file transmission rate from the base station substantially as against the above conventional schemes. In \cite{Maddah-Ali2014,Yu2017a}, an information theoretic approach is taken where the minimum rate required to fulfill requests from all the users is studied. The work in  \cite{Ji2016} extends similar results to D2D communication. \cite{Karamchandani} studies coded caching in a system with two layers of caches. In \cite{Pedarsani2016}, an online coded scheme is presented. These schemes have been widely studied under uniform popularity traffic and have been further extended to general random request case also (\cite{Ji2017,Niesen2017}).\\
\indent In the above coded caching  works, an important aspect of queuing at the server  has been ignored. The queuing delay can be the dominant component of the overall delay experienced by a user in a content delivery network. Work in~\cite{Rezaei2016a} addresses these issues. These authors propose a few queuing models for the cache aided coded multicasting schemes.

\indent A queue with multicasting and network coding is also studied in (\cite{Moghadam,Hou}) in a different setting where there is no \textit{finite} library of files to download from and each arriving packet is different and must be received by each receiver. \\
\indent As in \cite{Rezaei2016a} we also consider queuing delays at the server. Our major contributions are as follows:
	\begin{enumerate}
	\item We consider a new type of queue called the \textit{multicast queue} in which new requests for a file are \textit{merged} with the ones already in the queue. \textit{All} pending requests for a file get served with one transmission of a file. This exploits the broadcast nature of the wireless channel more effectively than in \cite{Rezaei2016a} and reflects the practical scenario more realistically. An immediate impact of this model is that, unlike in \cite{Rezaei2016a}, our queue at the server is \textit{always} stable for any request arrival rate from the users for any finite number of files.  Furthermore, the queue is quite different from the multicast queue studied  in \cite{Moghadam,Hou} because of the difference in traffic model considered. 
	\item We show existence and uniqueness of stationary distribution for the multicast queue. Next, we develop an approximate expression for the mean queuing delay for this queue. We also consider the case when there are errors in transmission. 
	\item Next, we combine our multicast queue with coded delivery schemes such as in \cite{Rezaei2016a} to further reduce the queuing delay. We show that our schemes are stable for all request arrival rates due to merging. We prove stationarity of this system and provide an approximate mean delay expression. We also show that our schemes provide lower mean delays than the schemes in \cite{Rezaei2016a} except at low request rates. At high request rates, with the LRU caches at the users, the multicast queue without coded delivery performs better than several schemes with coded delivery.
\end{enumerate}	 

Rest of the paper is organized as follows. Section \ref{sec:system_model} describes the system model. In  Section \ref{sec:perf_analysis}, we  derive an approximate expression for mean delay of the multicast queue. We use this to analyze a basic system with LRU caches and multicast queue at the server. We extend the analysis to the case where transmissions are not error free, perhaps due to fading. We also study the multicast queue with coded transmission in Section \ref{sec:coded_schemes}. We provide approximate mean delay for this system. In Section \ref{sec:compare_caching_schemes}, we compare the performance of the multicast queue with the coded multicast queue. 
Section \ref{sec:conclusion} concludes the paper. 
\section{System Model} \label{sec:system_model}
We consider a network with one server and multiple users (Figure \ref{fig:system_model}). The users share a common link with the server. This may model the downlink in a cellular network. Each user is equipped with a cache. The system has a total of $M <\infty$ files. The complete library of files is accessible to the server. The file library is denoted by $\mathcal{M} = \{1,2 \ldots M\}$. File $i$ is of size $F_i$ bits. The set of users is denoted by $\mathcal{L} = \{1,2 \ldots L\}$. The request process at each user is assumed to be IRM, i.e., request process for file $i$ at the user $j$ is a Poisson process with rate $\lambda_{ij}$ and is independent of the other request processes from it and all other users. Thus, user $j$ requests for file $i$  with probability $p_{ij} = \lambda_{ij}/\sum_i \lambda_{ij}$. If the popularity profile is Zipf with exponent $\alpha$, then $p_{ij} \propto i^{-\alpha}$. A user forwards its request to the server if the file  requested by it is not fully available in its cache.\\
\indent In the following, we will consider different schemes at the server for fulfilling requests from multiple users. The simplest and a well studied system is where all requests from different users are stored in a request queue at the server. The server transmits the requested files in first-in-first-out (FIFO) fashion over the shared link. We assume that the rate on the shared link is $R$ bits/second without errors (generalization to a channel with errors will also be considered). We call this system a \textit{FIFO} queue.\\
\begin{figure}[t!]
\centering
\includegraphics[scale=0.3]
{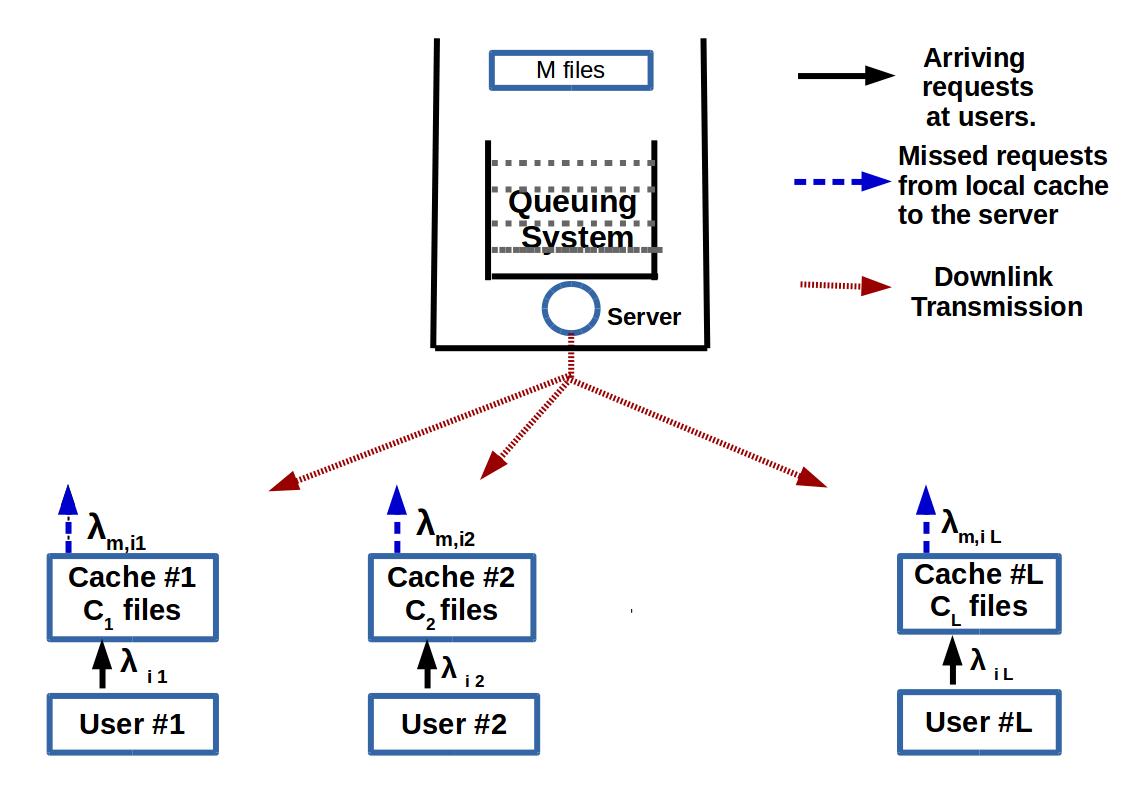}
\caption{System model has $L$ users with caches sharing a common link with one server with access to $M$ files and a queuing system with one or more queues at the server. User $j$ requests file $i$ according to a Poisson process of rate $\lambda_{ij}$.}
\label{fig:system_model}
\end{figure}
\indent Next, we consider the case where multicast nature of the channel is exploited: when the server transmits a file, all the users who have a pending request for that file receive the file and the corresponding requests are removed from the queue. We will describe this queue in more detail in Section \ref{sec:multicast_queue}. This queue has not been modeled previously. We will call this \textit{multicast} queue. We will see that this leads to significant improvement in performance over the FIFO queue. \\
\indent It has been demonstrated recently (\cite{Maddah-Ali2014}) that with cache enabled users and \textit{coded transmissions}, the multicast nature of the shared channel can be further exploited: it is possible to serve multiple users requesting different files simultaneously. Users retrieve the required files from the transmission from the server and the \textit{contents} of their own caches. We provide more details in Section \ref{sec:coded_schemes}.\\
\indent In the following, we theoretically study each of the above schemes and compare their performance.  
\section{Performance Analysis: FIFO and multicast queue}\label{sec:perf_analysis}
In Section \ref{subsec:fifo_queue}, we analyze the FIFO queue. Section \ref{sec:multicast_queue} describes the multicast queue in more detail and provides its performance analysis. Section \ref{sec:coded_schemes} presents several coded caching schemes and studies their queuing performance. 
\subsection{FIFO Queue}\label{subsec:fifo_queue}
We study the system without caches at the users. The transmission time $s_i$ for file $i$ is $F_i/R$ sec. Then, the request queue is an M/G/1 queue with arrival rate $\overline{\lambda} = \sum_{ij} \lambda_{ij}$ and i.i.d. service time $S$ with probability
\begin{equation}
\mathbb{P}[S = s_i] = \frac{\sum_{j} \lambda_{ij}}{\overline{\lambda}}.
\end{equation}
The queue is stable if $\rho \overset{\Delta}{=} \overline{\lambda} \mathbb{E}[S] < 1$. Under this condition, the queue will have a unique stationary distribution and the mean delay under stationarity is (\cite{Asmussen})
\begin{equation}
\mathbb{E}[D] = \frac{\rho}{1-\rho}  \left[ \frac{\mathbb{E}[S]}{2} + \frac{Var[S]}{2\mathbb{E}[S]} \right].
\end{equation}
If $\rho \geq 1$, the queue is unstable. It does not have a stationary distribution and the queue length will tend to infinity. \\
\indent Below, we study the multicast queue with caches at the users and also when transmission from the server has errors perhaps, because of fading channels. The same ideas can be used here to study the FIFO queue with caches and fading channel.

\subsection{Multicast Queue without Caches}
\label{sec:multicast_queue}
In this queue, we exploit the broadcast nature of the link in which all users with the pending request for a given file are serviced simultaneously. For simplicity, we first study the system when there are no caches at the users. The system with caches will be studied in Section \ref{sec:lru_multicast}. \\
\indent For this queue, the forwarded file request by a user is described by tuple $(i,j)$, where $i$ is the index of the file and $j$ is the user making this request. A request queued at the server is described by $(i,\mathbb{L}_i)$, where $i$ is the file index and $\mathbb{L}_i$ is the set of users who have requested the file $i$ and not yet  been served. When a new request for file $i$ arrives at the queue, the request is merged with its corresponding entry  in the queue if it already exists, i.e, the user list $\mathbb{L}_i$ is updated with the new users who have requested the file. If file $i$ is being served at that time, it is considered as a new request and is added to the tail of the queue. This is because the request has come from a user when part of the file has already been transmitted and hence the requesting user will not receive the full file from that transmission (we assume that a user does not receive a file being transmitted unless it has generated a request for this file).  
If at the time of arrival of request for file $i$ at the server, its request is not already in the queue, the new request is appended at the tail of the queue. The queue of requests is serviced one after the other from head to tail. The transmission of the file in the head of the queue serves all the users whose requests for that file have been merged.\\
\indent We study this queue in considerable detail as it is a basic natural model in wireless networks. Even though this queue has been exploited in previous scenarios, a detailed study of the same is missing in literature. By first studying this system, we isolate the benefits of multicasting alone, which does not require the extra overhead in terms of initial content placement and coding at the delivering time and then decoding at the receivers. Next, by studying the coded multicasting scheme, we compare the gains obtained via coded caching. Furthermore, the analysis used for the multicast queue will be useful in the analysis of the coded multicasting system as well.\\
\indent  We say that a request is of type $1$  if upon arrival there was no candidate for merger already in the queue. If the request was merged with an already existing request in the queue, we say it is of type $2$.\\
\subsubsection{Stationarity}
\indent We analyze this system and derive an approximate expression for the mean delay under stationarity. The maximum length of this queue is $M$ for all $\lambda_{ij}$ (because of merging of requests of each file).  Hence the queue is always stable. \\
\indent We consider the state $X_t$ of this system by associating a vector of tuples with each position in the queue at time $t$: $X_t = ((i_1,\mathbb{L}_{i_1}),(i_2,\mathbb{L}_{i_2}),\ldots,(i_n,\mathbb{L}_{i_n}))$ when $n$ is the queue length (including the file being served) and $i_k$ is the file index requested at position $k$ in the queue. Let $D_k$ denote the queuing delay of the $k^{th}$ arrival to the queue. We start with the following proposition. 
\begin{prop}\label{Th:BC_Prop}
For the multicast queue, $\{D_k\}$ is an aperiodic, regenerative process with finite mean regeneration interval and hence has an unique stationary distribution.  Also, starting from any initial distribution, $D_k$ converges in total variation to the stationary distribution.
\end{prop}
\begin{proof}
\indent Let the state of the queue just after the $n^{th}$ departure be $Y_n$. By IRM assumption, $\{Y_n\}$ is a finite state, irreducible discrete time Markov chain (DTMC) (\cite{Asmussen}). We claim that this is also aperiodic. To see this consider the state $Y_n=0$, i.e., when the system is empty just after the $n^{th}$ departure. The probability, $P(Y_{n+1}=0|Y_n=0)>0$ as there is a positive probability of just one arrival before the next departure. Thus the DTMC has a unique stationary distribution. \\
\indent The arrival epochs to the queue just after the epochs when $Y_n$ is $0$ are also the regeneration epochs for $\{D_k\}$. The regeneration lengths for $\{D_k\}$ also have finite mean length (although  it is more than that of $Y_k$ due to the merging effect) and can be shown to be aperiodic, as for $\{Y_k\}$. Thus, $\{D_k\}$ has a unique stationary distribution and starting from any initial conditions, $\{D_k\}$ converges to the stationary distribution in total variation. \\
\end{proof}   
\indent  Using above methods, we can show that $\{X_t\}$ also has a unique stationary distribution and starting from any initial distribution, the process converges in total variation to the stationary distribution.\\
\indent  The queue length is upper bounded by $M$ and the waiting times are upper bounded by $\Sigma_{i=1}^{M} s_i$. Thus, unlike for the FIFO queue, the multicast queue that is studied here has the distinct advantage of being always stable and has bounded delays. \\
\indent The dimension of the Markov chain $\{Y_n\}$ can be very large even for modest $M$ and $L$. Computing its stationary distribution can be prohibitive. Computing mean delay under stationarity is even harder. Next, we provide an approximation for the mean queuing delay using M/G/1 queue which is easy to compute.  \\
\subsubsection{M/G/1 approximation:}
Let $\overline{D_i}$ denote a random variable with the stationary distribution of queuing delay for requests of file $i$ and type $1$ across all users. We will use the approximation that $\mathbb{E}[\overline{D}_i]$ are same for all $i$ and we will denote it by $\mathbb{E}[\overline{D}]$ (we have checked via simulations that this  is a good approximation). Since all the requests for file $i$ that arrive during the delay of type 1 requests are merged and are considered as one single merged request for transmission, the \textit{effective} arrival rate of file $i$ into the queue is  
\begin{equation}\label{eq:effective_arr_rate}
\lambda'_i = \frac{\lambda_i}{1+\lambda_i \mathbb{E}[\overline{D}]},
\end{equation}
where $\lambda_i = \sum_j \lambda_{ij}$. The utilization factor for this queue is
\begin{equation}\label{eq:rho}
\rho  = \mathbb{E}[S] \sum_{i=1}^M{\lambda'_i}.
\end{equation}
We use the $M/G/1$ queue mean delay expression \cite{Asmussen} to approximate $\mathbb{E}[\overline{D}]$ as 
\begin{equation}\label{eq:PK}
\frac{\rho}{1-\rho} \left[ \frac{\mathbb{E}[S]}{2} + \frac{Var[S]}{2\mathbb{E}[S]} \right].
\end{equation}

\begin{figure}[t!]
\centering
\includegraphics[scale=0.4]
{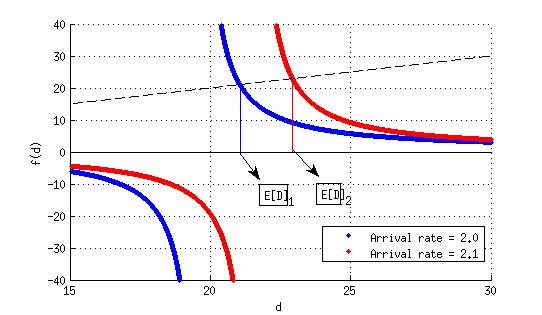}
\caption{Solution to M/G/1 approximation. $M = 100$, $s_i = 1s$. Zipf exponent 1. Plot shows $f(d)$ for total rate $= 2$ (blue) and $2.1$ (red). Corresponding fixed points are also indicated.}
\label{fig:fixed_point_graph}
\end{figure}
We use Eq (\ref{eq:effective_arr_rate}), (\ref{eq:rho}) and (\ref{eq:PK}) to solve for $\mathbb{E}[\overline{D}]$. From these equations, we can write (taking $\mathbb{E}[\overline{D}] = d$ for notational simplicity) the fixed point equation $f(d) = d$, where 
\begin{equation}
f(d) = \frac{\mathbb{E}[S] \sum_i \frac{\lambda_i}{1+\lambda_i d}}{ 1-\mathbb{E}[S] \sum_i \frac{\lambda_i}{1+\lambda_i d}}.
\end{equation}
From this equation, we can conclude the following:
\begin{enumerate}
\item $f$  is a continuous function of $d$ for all $d>d_0$ where $d_0$ satisfies $\mathbb{E}[S] \sum_i \frac{\lambda_i}{1+\lambda_i d_0} = 1$. For $d>d_0$, $f$ is strictly decreasing and as $d \rightarrow \infty $, $f(d) \rightarrow 0$. Also as $d \rightarrow d_0^+$, $f(d) \rightarrow \infty$. See Figure \ref{fig:fixed_point_graph}. The figure shows $f(d)$ for all $d$. The region where $d<d_0$ is of no interest to us. For any value of $\lambda_i$, $\mathbb{E}[S]$ and $Var[S]$, $f(d) = d$ has a unique solution in the region $d>d_0$. This solution is positive. This is the solution we consider for $\mathbb{E}[\overline{D}]$ for our system. At this $d$, $\rho < 1$.
\item If we increase any $\lambda_i$, keeping all the other parameters fixed in Eq (\ref{eq:rho}), $\rho$ increases and thus $f(d)$ increases for a fixed $d$ (see Figure \ref{fig:fixed_point_graph}). Therefore, the resulting fixed point $d=\mathbb{E}[\overline{D}]$ increases. If all $\lambda_i \rightarrow \infty$, then the fixed $d$, $f(d) \rightarrow \frac{\mathbb{E}[S] M/d}{1-\mathbb{E}[S] M/d} < \infty $. Thus, in the limit,  the fixed point will be $d$ at which $d = \frac{\mathbb{E}[S] M}{d-\mathbb{E}[S] M}$. This has the solution $d = \mathbb{E}[S] M \left(1+\sqrt{1+4/(M \mathbb{E}[S])}\right)/2$. For large $M$, $d \sim \mathbb{E}[S] M$. This will be an upper bound for all $\lambda_i$s.
\end{enumerate}
\indent Let $\widetilde{D}_i$ denote a random variable corresponding to queuing delay of file $i$ (includes both type $1$ and type $2$) under stationarity. We approximate $\mathbb{E}[\widetilde{D}_i]$ by 
\begin{equation}\label{eq:avg_delay}
\frac{\lambda'_i \mathbb{E}[D] + (\lambda_i - \lambda'_i) \mathbb{E}[D]/2}{\lambda_i}.
\end{equation}
This reflects the fact that $\lambda'_i/\lambda_i$ fraction of requests (type $1$) experience a delay of $\mathbb{E}[D]$. The requests that get merged with already existing ones in the queue (type $2$) experience a mean delay of $\mathbb{E}[D]/2$ (approximately).
 
\begin{figure}[t!]
\centering
\includegraphics[scale=0.115]
{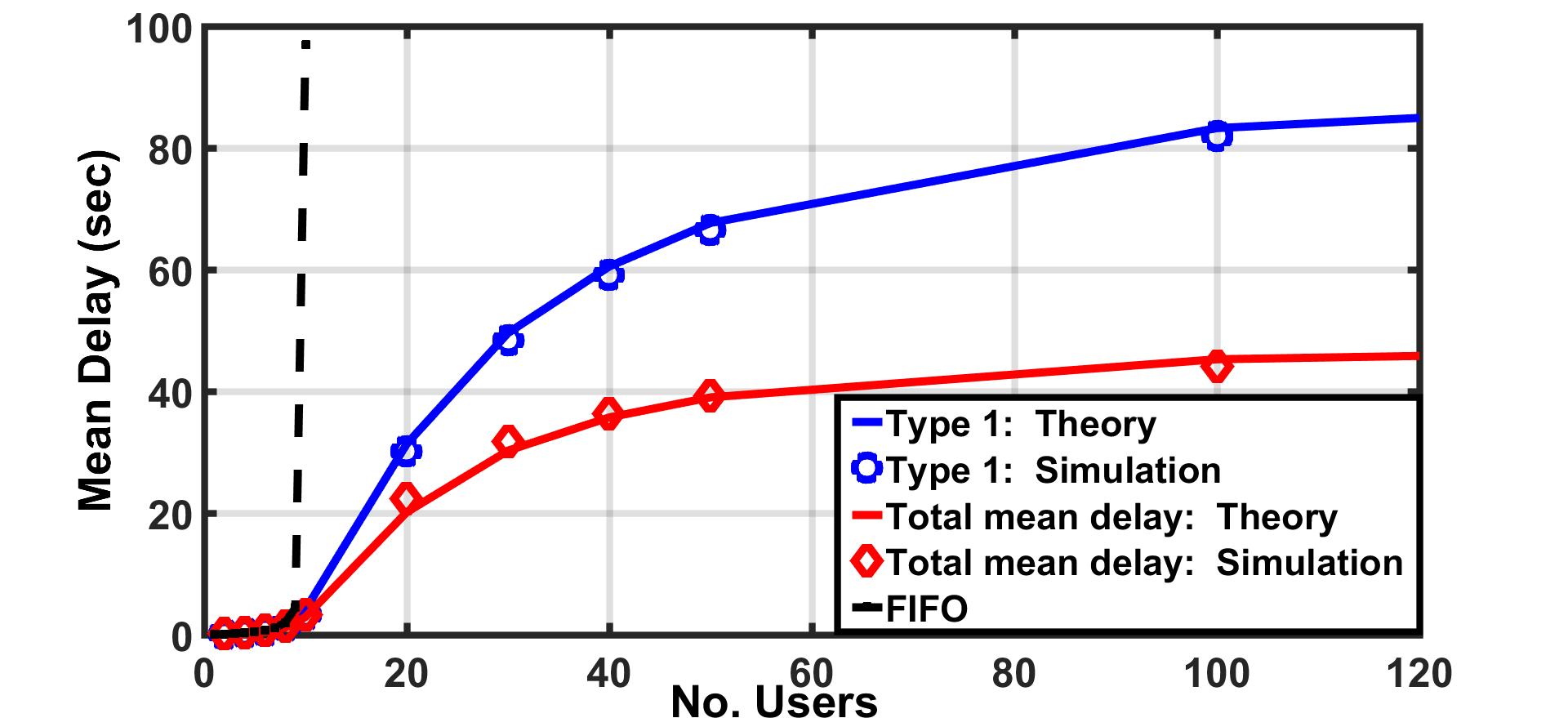}
\caption{Validity of M/G/1 approximation for the multicast queue. $M = 100$, $s_i = 1s$. user request rate = 0.1 request/sec and Zipf exponent 1.}
\label{fig:mg1_approximation}
\end{figure}

In Figure \ref{fig:mg1_approximation}, we compare the above approximation of the mean delay $\mathbb{E}[D]$ of the multicast queue with simulations. In this example, $M = 100$, $s_i = 1s$ for all files, popularity profile is Zipf with exponent $\alpha = 0.8$, $\lambda_{j} = 0.1$  files/second for all $j$. We see a good match. \\
\indent  In this figure, we also plot the mean delays for FIFO queue for comparison. FIFO queue becomes unstable as the number of users increases. However, the multicast queue always remains stable and has quite a low mean delay even at large arrival rates.\\
\indent  For the delay of $10s$, the maximum request arrival rate supported are approximately $9$ and $13.5$ files/sec for FIFO and the multicast queue respectively. For the mean delay of $15s$, the same are approximately $9.1$ and $16.2$ request/sec.

 \subsection{Multicast queue with LRU caches}
\label{sec:lru_multicast}
\indent Now, we extend our analysis to the multicast system with caches. All users use LRU at their caches. This is the most studied scheme for file replacement and has good performance (\cite{Martina2016}). Our analysis easily extends to other eviction policies also. The requests generated by a user that are not met at the local cache are forwarded to the server. At the server the requests are queued and served as described in Section   \ref{sec:multicast_queue}.
\subsubsection{Stationarity}
Consider the system state $X_t=\{X_t^0,X_t^1,...,X_t^L\}$, where, $X_t^0$ is the state of the multicast queue and $X_t^j,\ j=1,...,L$ are vectors with positions of each file in the cache $j,\ for\ j=1,...,L$. If a file is not there in the cache then the position corresponding to the file is set to 0. Let $D_k$ be the delay of the $k^{th}$ request arriving at the multicast queue.
\begin{prop}
The processes $\{X_t\}$ and $\{D_k\}$ have unique stationary distributions. Also, starting from any initial conditions, these processes converge in total variation to their stationary distributions.
\end{prop}
\begin{proof}

Denote by $Y_n$ the state $X_t$, just after the $n^{th}$ departure from the multicast queue. Due to IRM assumption, $\{Y_n\}$ is a finite state, aperiodic, irreducible Markov chain. Thus, $Y_n$ has a unique stationary distribution.\\
\indent Fix a state of $Y_n$ where the multicast queue is empty. The revisits to this state are regeneration epochs. Thus $\{Y_n\}$ has finite mean regeneration times. Now, consider the continuous time process $\{(X_t, R_t)\}$, where $R_t$ is the residual transmission time of the file being transmitted at the server at time $t$ This also has regeneration epochs at the regeneration epochs of $Y_n$ mentioned above and the regeneration times of $\{(X_t, R_t)\}$ have a finite mean due to Poisson arrivals. Thus $\{(X_t, R_t)\}$ has a unique stationary distribution. \\
\indent Since arrival epochs to the system form a Poisson process (IRM), by PASTA (\cite{Asmussen}) the arrivals see the same stationary distribution of the system as of $\{(X_t, R_t)\}$. Now, consider the arrivals at the server. These are the miss requests from the caches. Let $(Y'_n,R'_n)$ be the process at these epochs. The Palm distribution (\cite{Asmussen}) of $\{(X_t, R_t)\}$ at these epochs is the stationary distribution seen by these epochs. The queuing delay $D_k$ seen by the $k^{th}$ arrival at the server is $D_k=f(Y_k',R'_k)$ where $f$  is a deterministic function bounded by $\sum_{i=1}^{M}{s_i}$. Thus $D_k$ has a stationary distribution.\\
\indent These results also imply that starting from any initial conditions, $\{X_t\}$ and $\{D_k\}$ converge to their stationary distributions in total variation. 
\end{proof}
\subsubsection{Approximation of delay at the multicast queue} 
Since, miss requests are coming from LRU caches, if the IRM requests for file $i$ to user $j$ arrive at rate $\lambda_{ij}$, the missed traffic for file $i$ from user $j$, is given by 
\begin{equation}\label{eq:total_arrival_rate}
\lambda_{m,ij} = \lambda_{ij} P_{m}(i,j) 
\end{equation}
where $P_{m}(i,j)$ is the probability of miss for file $i$ at the cache of user $j$. Using Che's approximation in \cite{Martina2016}, $P_{m}(i,j) = e^{-\lambda_{ij} T_C}$,  
where $T_C$ is the Che's constant. Also, the arrival process at the multicast queue can often be approximated by a Poisson process (\cite{Martina2016}). Thus the effective arrival process at the server as in (\ref{eq:effective_arr_rate}) can be approximated by a Poisson process with rates
\begin{equation}\label{eq:Eff_lambda_for_coding}
\lambda_i'=\frac{\sum_{j=1}^{L}{\lambda_{m,ij}}}{1+E[D]\sum_{j=1}^{L}{\lambda_{m,ij}}},
\end{equation}\\
\begin{figure}[t!]
\centering
\includegraphics[scale = 0.13]
{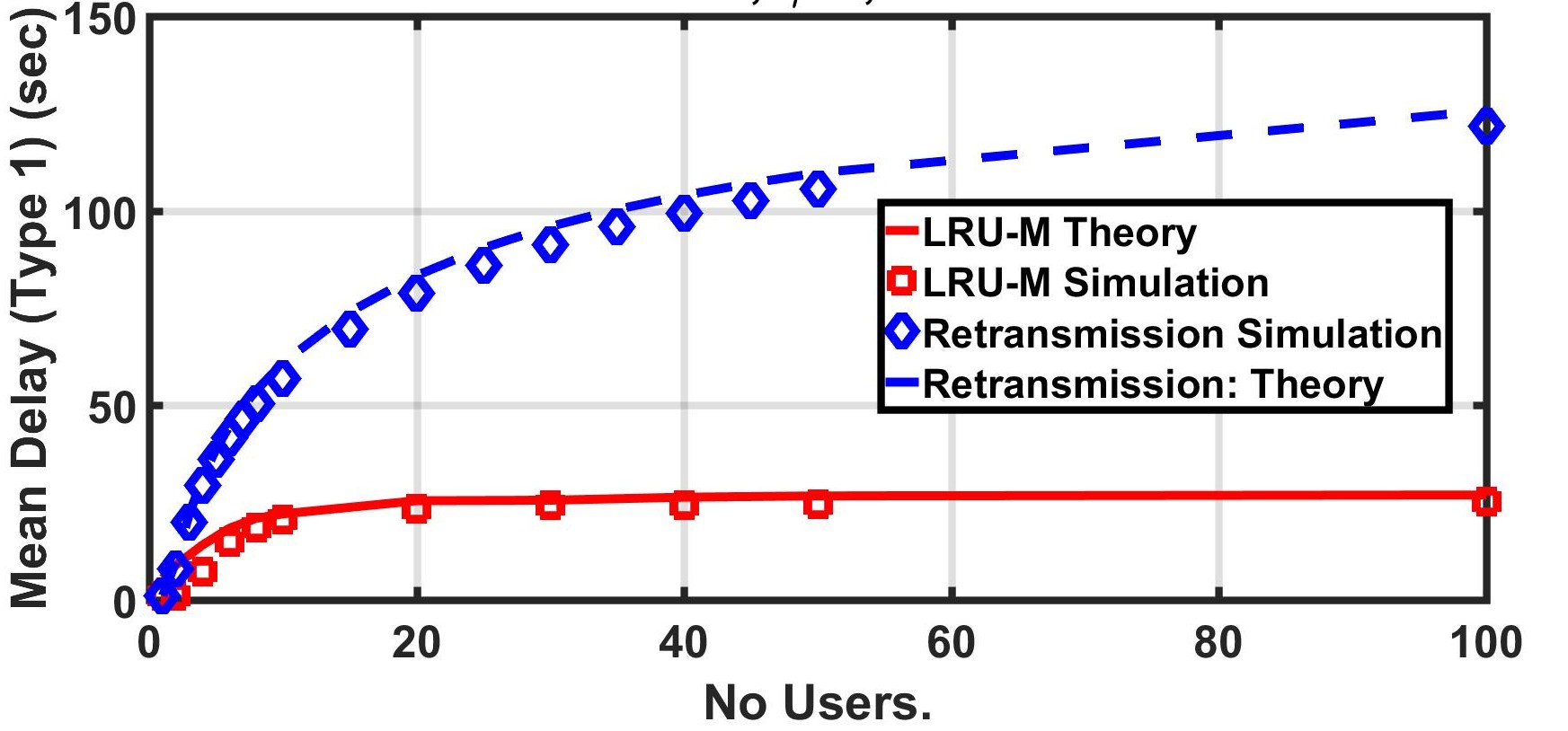}
\caption{Validity of theory for mean delay for the the multicast queue with transmission errors and LRU-multicast queue: $M=50$, $C=10$, $s_i = 1s$, error $1-r = 0.1$, user rate = 0.6 files/sec and Zipf exponent 1.}
\label{fig:Th_Vs_Sim_fading_lru-m}
\end{figure}
\indent Figure \ref{fig:Th_Vs_Sim_fading_lru-m} shows the comparison between theory and simulation. The parameters are $M=50$, $C=40$, $s_i = 1s$ for all $i$, arrival rate $= 0.6$ request/sec per user and file popularity distribution is Zipf with exponent equal to one. We have plotted the mean delay for the delays seen by all the requests coming at a user; if a request is met at the local cache, its delay is taken as zero. We see that they match well.\\
\indent If an eviction policy other than LRU is used, the only changes needed in the above approximation are in $P_m(i,j)$, which can be obtained from the literature (\cite{Martina2016}) for various policies. 
\subsection{Multicast queue with transmission errors}
\label{sec:fading_multicast}
Now, we consider a more realistic scenario where the channel from the server to the users experiences fading. This models a wireless cellular  network, where the Base Station (BS) is the server and the users are mobile customers. The fading processes for different users are  assumed independent of each other. The channel for a given user has block fading: the channel gain is same during the transmission of a file; it changes independently from one transmission to another with the same distribution. These are commonly made assumptions in literature \cite{Tse}.  \\
\indent We consider only the system without caches. We consider the multicast queue described in Sec \ref{sec:multicast_queue}. We assume that perfect knowledge of the channel gains from the BS to the mobile users is available at the BS before transmission. Also, the BS transmits at a constant power. \\
\indent We consider two schemes for transmission of a file: 
\begin{enumerate}
\item \textsc{Worst rate:} The server transmits at the rate corresponding to the user with the worst channel gain. After the transmission, all the intended users receive and decode the file transmitted without errors. 
\item \textsc{Retransmission, fixed rate:} There is an appropriately selected fixed rate $\overline{R}$ for transmission. The intended users whose current rate is greater than $\overline{R}$ will receive and decode correctly. All other users cannot decode correctly. The server retransmits the content to the remaining users till all the intended users receive and decode the content correctly. For this case, the BS does not need the channel gains; only a nack from the user not decoding correctly is sufficient. For selecting $\overline{R}$, we need the statistics of the fading channel.
\end{enumerate}
\begin{figure}[t]
\centering
\includegraphics[scale = 0.12]
{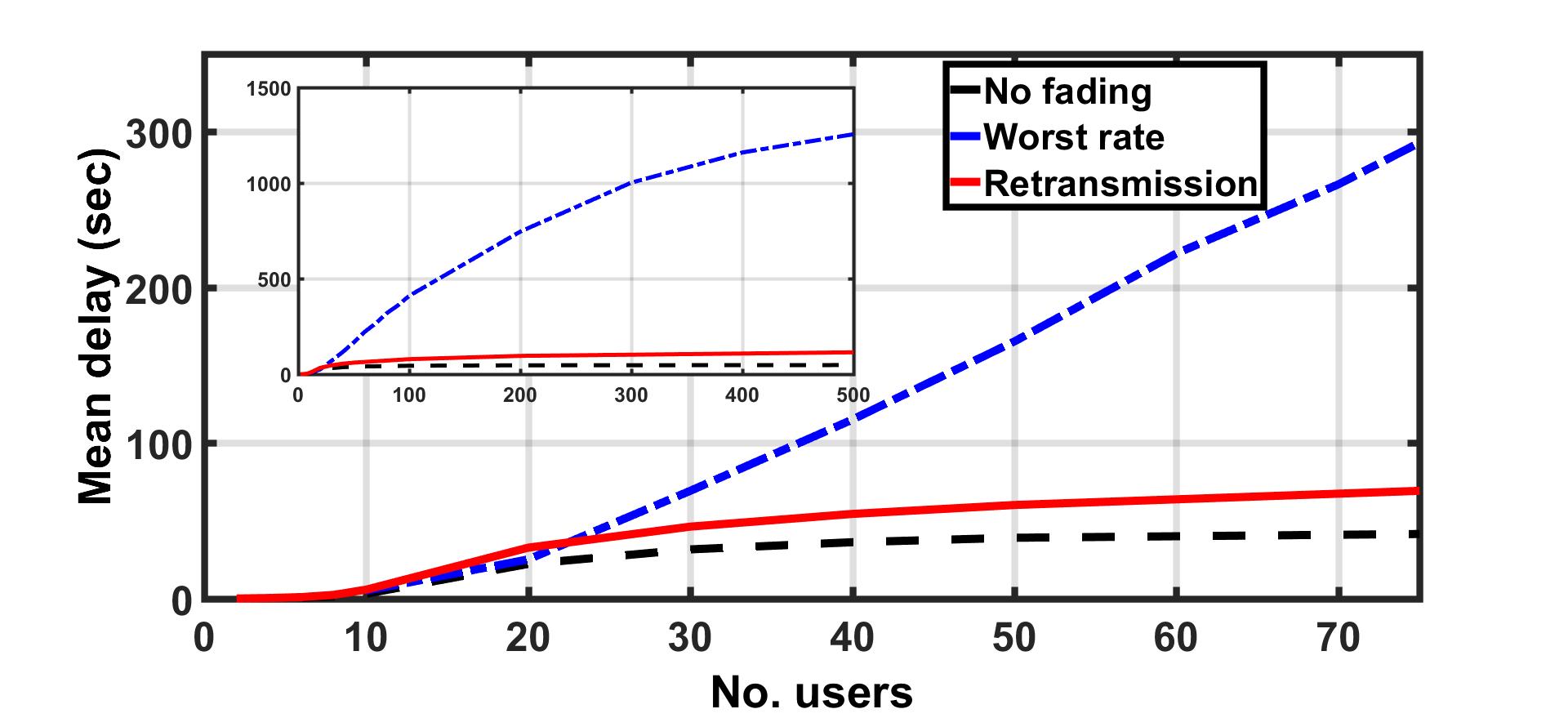}
\caption{Mean delay for the multicast queue with transmission errors.  $L=10$, user rate = $0.1$ request/sec, $M=100$, all file sizes with $\overline{R} = 1$ file/sec, $SNR = 10$, Bandwidth $= 10$ MHz, Zipf exponent $= 1$.}
\label{fig:multicast_fading_nocache_compare}
\end{figure}

We compare the performance of the two schemes by an example. The channel is additive white Gaussian noise channel. The system parameters are: $L=10$, user rate = $0.1$ request/sec, $M=100$, file sizes are same with $\overline{R} = 1$ file/sec, $SNR = 10dB$, $Bandwidth = 10 MHz$. Rayleigh fading is assumed with channel gain $H$ and $H^2 \sim exp(1)$. We plot mean delay of the queue in Figure \ref{fig:multicast_fading_nocache_compare}. For comparison, we also plot the mean delay when there is no fading. The retransmission scheme outperforms the worst user rate scheme substantially. Thus, in the following we study only the retransmission scheme. 
\subsubsection{M/G/1 approximation}
For the theory, we again use M/G/1 approximation of Sec \ref{sec:multicast_queue}. We need to compute the mean and variance of service times. For this, we need the distribution of the number of users merged who have requested the file being serviced. On the average a request of type $1$ for file $i$ from a user $j$ spends time equal to $E[D]$ in the multicast queue from the time it enters the queue till its service begins. The probability that any other user $k$ makes a request (type $2$) for the same file is $q_{ik} = 1 - exp(-\lambda_{ik}E[D])$. The request process from each user is assumed to be independent. From this, we can get the distribution of the number of users merged corresponding to a given request being serviced. Let the type 1 request corresponding to the file $i$ in service be from user $j$. The number of users is given by $1+\sum_{k \in \mathbb{L} \setminus j} X_k$, where $X_k$ has Bernoulli distribution with the parameter $q_{ik}$. Also, the probability that type 1 request was generated by user $j$ is $\lambda_{ij}/\sum_k \lambda_{ik}$. \\
\indent Let $r_j$ be the probability that user $j$ receives and decodes a transmission correctly. For the AWGN channel with transmit power $P$, channel gain $H$ and noise variance $\sigma^2$, $r_j= \mathbb{P}[log(1+HP/\sigma^2)/2 \geq \overline{R}]$. The number of transmissions required for successful reception is geometrically distributed with parameter $r_j$. If $\mathbb{L}_i$ denotes the set of users requesting the file $i$ that is being transmitted, the total number of transmissions $N_{t,i}$ required to serve all the users in $\mathbb{L}_i$ has the distribution
\begin{equation} 
P[N_{t,i} \leq n] =  \prod_{j\in\mathbb{L}_i}\ \   (\sum_{m=1}^{n} r_j(1-r_j)^{m-1}).
\end{equation}
From this and the distribution of number of users in $\mathbb{L}_i$, we can compute $\mathbb{E}[S]$ and $Var[S]$ in Eq (\ref{eq:PK}). The resulting set of equations can be solved numerically. \\
\indent In Figure \ref{fig:Th_Vs_Sim_fading_lru-m}, we compare $\mathbb{E}[D]$ computed theoretically with simulations. The parameters are $M=50$ files, all files are of the same size with transmission rate of each file = $1$ file/sec, error rate $(1-r_j) = 0.1$ for all $j$ and arrival rate = $0.6$ request/sec per user. Popularity profile of the files follows Zipf distribution with exponent 1. We see that the theory matches well with simulations. \\
\indent If the users also have caches, then the above analysis can be extended to this system by replacing the request rates $\lambda_{ij}$ by (\ref{eq:total_arrival_rate}) and (\ref{eq:Eff_lambda_for_coding}).
\section{Coded multicasting schemes with queueing}\label{sec:coded_schemes}
\label{sec:PCS_Merge}
In this section, we consider the system when users are equipped with caches and use coded multicasting for transmission. We assume all files have the same size ($F_i = F$ for all $i$). Each user is assumed to have a cache of size equal to $C$ files. These assumptions are being made because we are using the coding scheme in \cite{Maddah-Ali2014}, but can be relaxed.\\
\indent The caches at the users are populated with parts of the file as in \cite{Maddah-Ali2014}. Since, each user stores a certain part of all files, no request can be completely met locally. Hence, all requests are sent to the server. At the server, a separate queue for requests from each user is maintained. Within each user queue, requests corresponding to the same file are merged as in the multicast queue. During transmission, the head of the line requests of all the queues are considered for coding and the delivery scheme in \cite{Maddah-Ali2014} is used. Some of the queues can be empty. For such queues, a dummy request (for an arbitrary file) is assumed to be present.  Note that assuming dummy requests when some queues are empty is quite wasteful. In the next section, we will consider policies which take care of this. We call this scheme \textit{Partition Coded Scheme With Merging} (PCS-M). We will show in Sec \ref{sec:compare_caching_schemes} that this scheme outperforms the multicast queue substantially at low arrival rates. However, at high rates the probability that most of the users are merged in each file request for the multicast case is high and hence the multicasting gains take over. \\
\indent In the following, we study this system theoretically. 
\subsection{Stationarity}
Let the system state, $X_t=\{X_t^1, ..., X_t^L\}$ where $X_t^j$ is a vector representing  the index of each file in the queue of user $j$ at the BS. Let $D_k$ be the delay of the $k^{th}$ arrival to the BS.
\begin{prop}
For the system with caches and coding, the queuing delay process $\{D_k\}$  is an aperiodic, regenerative process with a finite mean regeneration interval and hence has a unique stationary distribution.
\end{prop}
\begin{proof}
Let $\{Y_k\}$ be the state of the system just after the $k^{th}$ departure from the BS. $\{Y_k\}$ is an irreducible, finite state, aperiodic DTMC. Thus it has an unique stationary distribution. Further, we note that the epochs where $Y_k=0$, i.e., all the queues are empty, are regeneration epochs for the process $\{Y_k\}$. The finiteness of the mean of the regeneration intervals holds. \\
\indent Now, since all the IRM arrivals at the users are forwarded to the BS, the forwarded requests also form a Poisson process. The system states seen by these arrivals also have  the regeneration epochs which see the system empty (the arrivals just after $Y_k = 0$, see empty system). Thus, $D_k$ have a unique stationary distribution.   
\end{proof}
\subsection{Approximation}
\indent To compute the mean queuing delay under stationarity, we consider each user queue separately. We use the M/G/1 approximation from Sec \ref{sec:multicast_queue}. The effective arrival rate for file $i$ into the queue for user $j$ is 
\begin{equation}
\lambda'_{ij} = \frac{\lambda_{ij}}{1+\lambda_{ij} \mathbb{E}[D_j]} \mbox{ for } j = 1,2, \ldots L \mbox{ and } i = 1,2, \ldots M
\end{equation}
where $\mathbb{E}[D_j]$ is the mean queuing delay for type 1 requests from user $j$. The service time is equal to the time it takes to transmit coded packets for all the $L$ users. The transmission rate of the channel is $1$ file/sec. In the scheme considered, the service times are deterministic and given by (\cite{Maddah-Ali2014})
\begin{equation}
s = \frac{L(1-C/M)}{1+CL/M}.
\end{equation}
We use these in Eq (\ref{eq:rho}) and Eq (\ref{eq:PK}) to compute $\mathbb{E}[D_j]$ (we will have $\mathbb{E}[D_j]$ in the place of $\mathbb{E}[D]$ in Eq (\ref{eq:PK})) for each queue).\\
\begin{figure}[t!]
\centering
\includegraphics[height=4.5cm, width=8.5cm]
{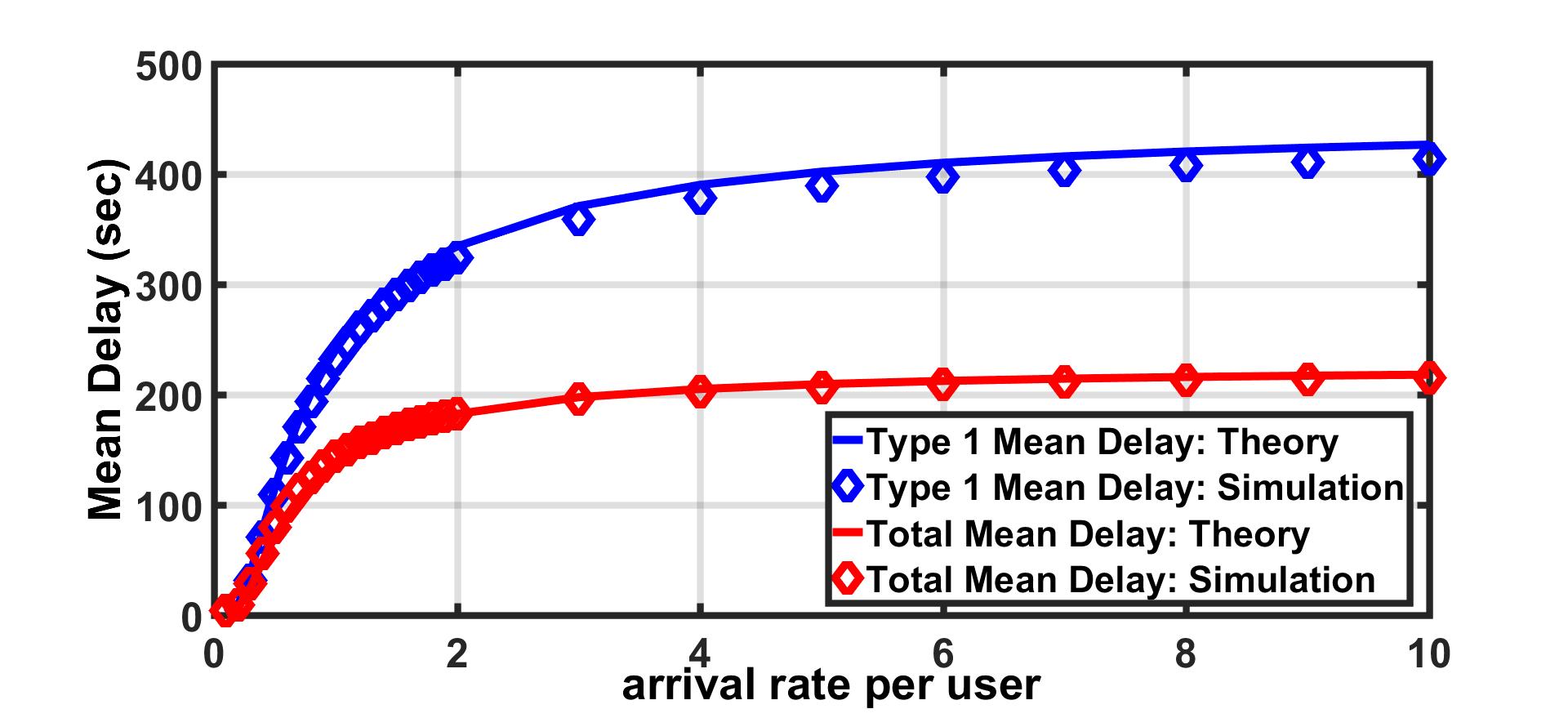}
\caption{M/G/1 approximation for the coded multicasting queue at user \#1. Transmission rate = 1 file/sec, $M = 100$, $C = 10$ and $L = 10$. Zipf Exponent 1.}
\label{fig:mg1_approximation_coded_multicast}
\vspace{-5mm}
\end{figure}
\indent Figure \ref{fig:mg1_approximation_coded_multicast} shows the mean queuing delay for user \#1 as the request rate is varied. The system parameters are: transmission rate = 1 file/sec, $M = 100$, $C = 10$ and $L = 10$. File popularity is Zipf distribution with exponent 1. We notice that M/G/1 approximation agrees well with the simulations. \\
\indent These results can be extended to the fading channel, as for the multicast queue in Section \ref{sec:fading_multicast}.
\section{Comparison of coded and uncoded caching schemes}
\label{sec:compare_caching_schemes}
In this section, we compare the mean delays of several schemes via simulations. These schemes include the above studied schemes, some from \cite{Rezaei2016a} and a few additional variants of the above schemes.
\begin{enumerate}
\item \textit{LRU with merging and multicast (LRU-M):} This scheme is described in Section \ref{sec:lru_multicast}.
\item \textit{Partition coded scheme with merging (PCS-M):} This scheme is described in Section \ref{sec:coded_schemes}. 
\item \textit{Modified PCS with merging (MPCS-M):} This is an improvement over PCS-M where the dummy requests are not used but delivery of the head of the line requests in all non-empty queues is performed as in WPCS in \cite{Rezaei2016a}. 
\item \textit{Coded delivery LRU scheme (CDLS):} This the same as CDLS in \cite{Rezaei2016a}. The users have LRU caches. There is a separate queue of pending requests for each user at the server. At each transmission time, head of the line of queues of each user is checked for coding opportunity. If there is no coding opportunity, the request with the largest waiting time is fulfilled.
\item \textit{Coded delivery LRU scheme with merging (CDLS-M):} In this scheme, we include merging of requests for the same file in each user queue. The delivery procedure is the same as in CDLS.
\item \textit{Uncoded pre-fetch optimal with merging (UPO-M):} This scheme uses the coded-caching scheme in \cite{Yu2017a}. Also, there is merging of requests for the same file in the queue from each user.
\item \textit{LRU with merging and coded delivery (LRU-CM):} In this scheme, users have LRU caches. The missed requests from all the users are queued in one multicast queue at the server.  During delivery, coding opportunity between the requests in head of the line and the one immediately after it is searched for. The best of the coded or uncoded delivery in terms of the number of users served is chosen. This scheme should work at least as well as the LRU-M.
\end{enumerate}

\begin{figure}[!t]
\centering
\includegraphics[scale = 0.14]
{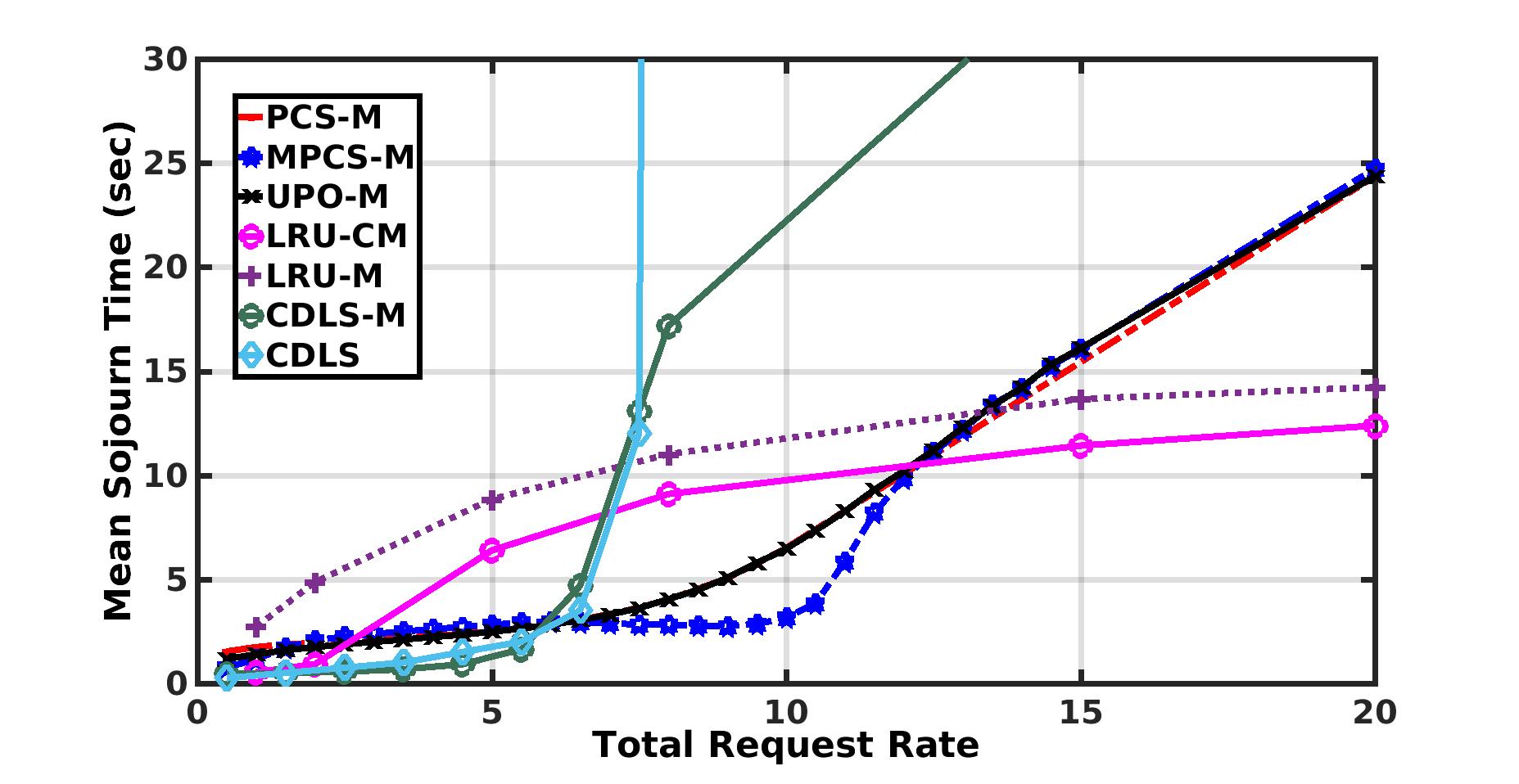}
\caption{Mean delay for various schemes. $M=100$, $C=40$, $L=10$, transmission rate $= 1$ file/sec and Zipf exponent $1$.}
\label{fig:compare_coded_schemes_1}
\end{figure}
  
\begin{figure}[!t]
\centering
\includegraphics[scale = 0.14]
{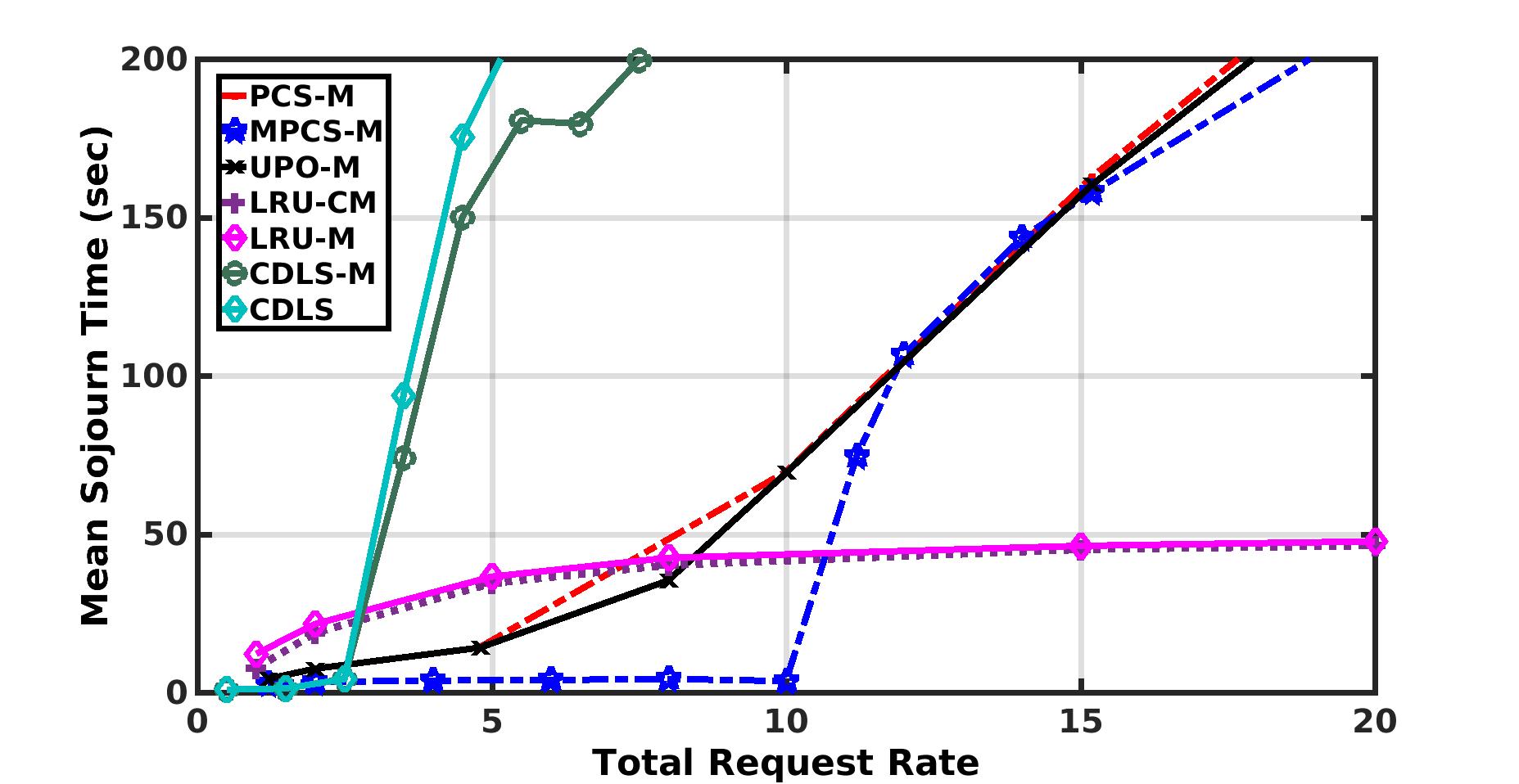}
\caption{Mean delay for various schemes. $M=200$, $C=25$, $L=40$, transmission rate $= 1$ file/sec and Zipf exponent $1$.}
\label{fig:compare_coded_schemes_2}
\end{figure}

\indent Figure \ref{fig:compare_coded_schemes_1} shows the mean sojourn times. This includes the transmission time in the BS queue, because now the transmission times of the coded and uncoded transmissions will be different for different schemes. Also, if a request is satisfied at a local cache, its sojourn time is zero. The system parameters are $M=100$, $C=40$, $L=10$, files are of the same size with transmission rate $= 1$ file/sec and the popularity profile is Zipf with exponent $1$. \\
\indent The schemes UPO-M and PCS-M have similar mean delay. MPCS-M is slightly better for low arrival rate, but, eventually has almost the same performance as the UPO-M and the PCS-Merge. We see that  CDLS at low rates is somewhat better than other schemes but becomes unstable as total arrival rate at the server increases.  All other schemes are stable because merging of requests establishes an upper bound on the number of requests in the queue. However, the mean delay for the above three schemes using coded multicasting is much lower at low arrival rates than LRU-M but as the rates increase, LRU-M and LRU-CM perform better. For low rates, performance  of CDLS-M and CDLS is similar. However, CDLS-M does not become unstable as arrival rates increase, but, performs worse than other schemes. \\
\indent If we consider a mean delay limit of $15s$, the  approximate maximum arrival rates are $7.5$ for CDLS, $7.6$ for CDLS-M, $14.8$ for PCS-M, MPCS-M and UPO-M. The maximum arrival rates is greater than $20$ for LRU-M and LRU-CM.\\
\indent In Figure \ref{fig:compare_coded_schemes_2}, we again compare the above schemes for $M=200$, $C=25$, $L=40$, transmission rate $= 1$ file/sec and file popularity follows Zipf distribution with exponent $1$. 
Here also, we see a similar comparison of different schemes as in Figure \ref{fig:compare_coded_schemes_1}.
\section{Conclusion}
\label{sec:conclusion}
We have studied a system where a server transmits files to multiple users over a shared wireless channel. The users have caches. New queuing models of the system are obtained when broadcast nature of the channel is fully exploited and when in addition, network  coding is also used. Our models have the distinct advantage that the system is always stable for all request arrival rates. Mean delay of the two systems are obtained theoretically and compared. For our models, we have observed that at lower arrival rates, the coded multicasting schemes provide less mean delay, but, at higher rates the uncoded multicasting eventually takes over.
\bibliographystyle{IEEEtran}
\bibliography{library.bib}
\end{document}